\documentclass[reqno,11pt]{amsart}
\usepackage{amsmath, latexsym, amsfonts, amssymb, amsthm, amscd}
\usepackage[utf8]{inputenc}
\usepackage{graphics,epsf,psfrag,epsfig}

\newcommand{\Ey}{\mathbb{E}^{Y,\rho'} }

\numberwithin{equation}{section}

\newtheorem{theorem}{Theorem}[section]
\newtheorem{lemma}[theorem]{Lemma}
\newtheorem{proposition}[theorem]{Proposition}
\newtheorem{cor}[theorem]{Corollary}
\newtheorem{rem}[theorem]{Remark}

\renewcommand{\ge}{\geq}
\renewcommand{\le}{\leq}
\newcommand{\ind}{\mathbf{1}}

\newcommand{\R}{\mathbb{R}}
\newcommand{\Z}{\mathbb{Z}}

\renewcommand{\tilde}{\widetilde}

\newcommand{\cA}{{\ensuremath{\mathcal A}} }

\newcommand{\cT}{{\ensuremath{\mathcal T}} }

\newcommand{\bP}{{\ensuremath{\mathbf P}} }
\newcommand{\bE}{{\ensuremath{\mathbf E}} }

\DeclareMathSymbol{\leqslant}{\mathalpha}{AMSa}{"36} 
\DeclareMathSymbol{\geqslant}{\mathalpha}{AMSa}{"3E} 
\DeclareMathSymbol{\eset}{\mathalpha}{AMSb}{"3F}     
\renewcommand{\leq}{\;\leqslant\;}                   
\renewcommand{\geq}{\;\geqslant\;}                   
\newcommand{\dd}{\,\text{\rm d}}             

\newcommand{\bbE}{{\ensuremath{\mathbb E}} }

\newcommand{\bbP}{{\ensuremath{\mathbb P}} }

\newcommand{\bbR}{{\ensuremath{\mathbb R}} }

\newcommand{\bbZ}{{\ensuremath{\mathbb Z}} }

\newcommand{\gb}{\beta}
\newcommand{\gd}{\delta}
\newcommand{\gep}{\varepsilon}       

\newcommand{\gz}{\zeta}

\newcommand{\gk}{\kappa}

\newcommand{\gO}{\Omega}
\newcommand{\gl}{\lambda}

\makeatletter
\def\captionfont@{\footnotesize}
\def\captionheadfont@{\scshape}

\long\def\@makecaption#1#2{%
  \vspace{2mm}
  \setbox\@tempboxa\vbox{\color@setgroup
    \advance\hsize-6pc\noindent
    \captionfont@\captionheadfont@#1\@xp\@ifnotempty\@xp
        {\@cdr#2\@nil}{.\captionfont@\upshape\enspace#2}%
    \unskip\kern-6pc\par
    \global\setbox\@ne\lastbox\color@endgroup}%
  \ifhbox\@ne 
    \setbox\@ne\hbox{\unhbox\@ne\unskip\unskip\unpenalty\unkern}%
  \fi
  \ifdim\wd\@tempboxa=\z@ 
    \setbox\@ne\hbox to\columnwidth{\hss\kern-6pc\box\@ne\hss}%
  \else 
    \setbox\@ne\vbox{\unvbox\@tempboxa\parskip\z@skip
        \noindent\unhbox\@ne\advance\hsize-6pc\par}%
\fi
  \ifnum\@tempcnta<64 
    \addvspace\abovecaptionskip
    \moveright 3pc\box\@ne
  \else 
    \moveright 3pc\box\@ne
    \nobreak
    \vskip\belowcaptionskip
  \fi
\relax
}
\makeatother
\def\writefig#1 #2 #3 {\rlap{\kern #1 truecm
\raise #2 truecm \hbox{#3}}}

\newcommand{\tf}{\textsc{f}}

\renewcommand{\a}{\mathrm{ann}}
\newcommand{\p}{\mathrm{pin}}

\title[Effect of disorder for Randow Walk Pinning Model]{The effect of disorder on the free-energy for the Random Walk Pinning Model: smoothing of the phase transition and low temperature asymptotics}

\begin{document}

\author{Quentin Berger}
\address{
Laboratoire de Physique, ENS Lyon,  Universit\'e de Lyon, 46 All\'ee d'Italie, 
69364 Lyon, France
}
\email{quentin.berger@ens.fr}

\author{Hubert Lacoin}
\address{Università degli Studi “Roma Tre”,
Largo San Leonardo Murialdo,
00146 Roma}
\email{lacoin@math.jussieu.fr}

\begin{abstract}
  We consider the continuous time version of the {\sl Random Walk Pinning Model} (RWPM), studied in \cite{BT, BS08, BS09}.
  Given a fixed realization of a random walk $Y$ on $\Z^d$ with jump rate $\rho$ (that plays the role of the random medium),
  we modify the law of a random walk $X$ on $\Z^d$ with jump rate $1$ by reweighting the paths, giving an energy reward
  proportional to the intersection time $L_t(X,Y)=\int_0^t \ind_{X_s=Y_s}\dd s$:
  the weight of the path under the new measure is $\exp(\gb L_t(X,Y))$, $\gb \in \bbR$.
  As $\beta$ increases, the system exhibits a delocalization/localization transition:
  there is a critical value $\beta_c$, such that if $\gb>\gb_c$
  the two walks stick together for almost-all $Y$ realizations.
  A natural question is that of disorder relevance, that is whether the
  {\sl quenched} and {\sl annealed} systems have the same behavior. 
  In this paper we investigate how the disorder modifies the shape of the free energy curve:
  (1) We prove that, in dimension $d\geq 3$, the presence of disorder makes
  the phase transition at least of second order. This, in dimension $d\ge 4$, contrasts with the fact that the phase
  transition of the annealed system is of first order. (2) In any dimension, we prove that disorder
  modifies the low temperature asymptotic of the free energy.
\end{abstract}

\maketitle

\section{Model and results}

\subsection{The random walk pinning model}
\label{fw}

Let $X=(X_s)_{s\geq 0}$ and ${Y=(Y_s)_{s\geq 0}}$ be two independent
continuous time random walks on $\Z^d$, $d\ge 1$, starting from $0$, with jump rates
$1$ and $\rho\geq 0$ respectively, and which have identical irreducible symmetric jump probability kernels.
We also make the assumption that the increments $X$ and $Y$ on $\Z^d$ have finite second moments. We denote by $\bbP^X$, $\bbP^{Y,\rho}$ 
the associated probability laws.

For $\gb\in\R$ (when $\gb\ge 0$, it should be considered as the inverse temperature),
$t\in\R_{+}$,
and for a fixed realization of $Y$,
we define a Gibbs transformation of the path measure $\bbP^X$:
the polymer path measure $\mu_{t,\gb}^{Y,\p}$. It is absolutely continuous with
respect to $\bbP^X$, and its Radon-Nikodym derivative is given by
\begin{equation}
\label{eq:defmu}
  \frac{\dd \mu_{t,\gb}^{Y,\p}}{\dd \bbP^X} (X) = \frac{e^{\gb L_t(X,Y)}\; \ind_{\{X_t=Y_t\}}}{Z_{t,\gb}^{Y,\p}},
\end{equation}
where $L_t(X,Y):=\int_0^t \ind_{\{X_s=Y_s\}} \dd s$ is the intersection time between $X$ and $Y$,
and 
\begin{equation}
Z_{t,\gb}^{Y,\p}:= \bbE^X\left[e^{\gb L_t(X,Y)}\; \ind_{\{X_t=Y_t\}}\right]
\end{equation}
is the so-called partition function of the system: it is the factor that normalizes $\mu_{t,\gb}^{Y,\p}$ to a probability law.
One can think of $\mu_{t,\gb}^{Y,\p}$ as a measure under which the walk $X$ is given an energy reward $\gb$ for staying in touch with $Y$.
The superscript ``$\p$'' refers to the fact that $X$ is constrained to be {\sl pinned} to $Y$ at its end point $Y_t$. This constrain is taken 
for practical reasons (see below)
and it can be removed without affecting the main features of the model.

Given a trajectory $Y=(Y_s)_{s\ge0}$, we also define the partition function along a time interval $[t_1,t_2]$ as 
\begin{equation}
  Z^{Y,\p}_{[t_1,t_2],\gb}:= Z_{t_2-t_1,\gb}^{\theta_{t_1} Y,\p},
\end{equation}
where $\theta_t Y:= (Y_{s+t}-Y_t)_{s\ge 0}$ ($\theta_t$ is the shift operator along time, it preserves the law of $Y$).

\medskip 
We give now a physical interpretation to this model:
The graph of the random walk $(s,X_s)_{s\in[0,t]}$ models a $1$-dimensional polymer chain living in a $(d+1)$-dimensional space
interacting with a random defect line $(s,Y_s)_{s\in [0,t]}$. The Gibbs measure $\mu_{t,\gb}^{Y,\p}$ is the measure that gives the law of the polymer configuration 
$(s,X_s)_{s\in[0,t]}$ at inverse temperature $\gb$, given a fixed realization of the defect line $Y$.
We are interested in the typical behavior of large systems, that is with large $t$.
At low temperature (large $\gb$), the interaction energy dominates the entropy and the polymer sticks to $Y$, and it is said to be localized.
At high temperature (small $\gb$), the entropy dominates and the polymer wanders away from $Y$, and it is said to be delocalized.
The aim of this paper is to get a better understanding of the phase transition in $\gb$ between the delocalized and localized phase, that is the
behavior at the critical temperature, and of the low temperature behavior of the polymer.

\medskip

As it is shown later in the introduction this is natural to compare
this model with a simpler and exactly solvable model where the random defect line $(s,Y_s)_{s\in [0,t]}$, is replaced by a deterministic one $[0,t]\times \{0\}$.

More physical motivation for the model are given in the introduction of \cite{BS08}.


\begin{rem}[Superadditivity]  \rm
One fundamental property of the \textsl{pinned} partition function,
is the stochastic superadditivity of $\log Z_{t,\gb}^{Y,\p}$. Indeed, for any $0\leq s \leq t$ and $\gb\in\R$,
\begin{equation}
 Z_{t,\gb}^{Y,\p} \geq \bbE^X\left[\ind_{\{X_s=Y_s\}} e^{\gb L_t(X,Y)} \ind_{\{X_t=Y_t\}} \right] = Z_{s,\gb}^{Y,\p} Z_{[s,t],\gb}^{Y,\p}.
\end{equation}
This remark applies also to the partition function along any time interval:
\begin{equation}\label{superaddi}
Z_{[u,w],\gb}^{Y,\p} \geq Z_{[u,v],\gb}^{Y,\p} Z_{[v,w],\gb}^{Y,\p},\ \text{ for any } u\leq v\leq w.
\end{equation}
\label{rem:superad}
\end{rem}
This crucial property allows (with some additional effort) to prove the existence of the Lyapunov exponent of free energy:

\begin{proposition}[from \cite{BS08}  Thm.1.1 and Cor.1.3]

 The limit
\begin{equation}
 \tf(\gb,\rho):=\lim_{t\to \infty} \frac{1}{t}\log Z_{t,\gb}^{Y,\p}
\end{equation}
exists and is non-random $\bbP^{Y,\rho}$ almost surely. We call it the {\sl quenched} free energy. In addition we have
\begin{equation}\label{sup}
 \lim_{t\to \infty} \frac{1}{t}\bbE^{Y,\rho}\left[\log Z_{t,\gb}^{Y,\p}\right]
= \sup_{t>0}\frac{1}{t}\bbE^{Y,\rho}\left[\log Z^{Y,\p}_{t,\gb}\right]= \tf(\gb,\rho).
\end{equation}
Moreover $\gb\mapsto \tf(\gb,\rho)$ is non-decreasing and non-negative so that there exists a value $\gb_c(\rho)$ such that
\begin{center}
 $\tf(\gb,\rho)>0\Leftrightarrow \gb>\gb_c(\rho)$.
\end{center}
The free-energy of the averaged system is called the {\sl annealed} free energy and is defined by
\begin{equation}\begin{split}
 \tf^{\a}(\gb,\rho):=\lim_{t\to\infty} \frac{1}{t}\log \bbE^{Y, \rho} \left[Z_{t,\gb}^{Y,\p}\right],\\
 \gb_c^{\a}(\rho):=\inf\left\{\gb \ | \ \tf^{\a}(\gb,\rho)>0\right\}.
\end{split}\end{equation}
We have by Jensen inequality that $\tf(\gb,\rho)\le  \tf^{\a}(\gb,\rho)$, and $\gb_c(\rho)\ge \gb_c^{\a}(\rho)$.
\end{proposition}

\begin{rem}\rm
The critical value $\gb_c(\rho)$ identifies the phase transition between the localized and the delocalized phase.
The fact that $X$ sticks to $Y$ when $\gb>\gb_c(\rho)$ can be seen from the fact that
\begin{equation}
 \frac{\partial }{\partial \gb }\log Z_{t,\gb}^{Y,\p} = \mu_{t,\gb}^{Y,\p} (L_t(X,Y)),
\end{equation}
so that, using convexity and passing to the limit
\begin{equation}
\lim_{t\to \infty}\frac{1}{t}\mu_{t,\gb}^{Y,\p} (L_t(X,Y))=\tf'(\gb,\rho)
\end{equation}
whenever the right-hand side exists. This shows that $L_t(X,Y)$ is asymptotically of order $t$ in the localized phase.
\end{rem}




\subsection{The pure model}

In order to be able to compare the quenched free energy curve with the annealed one, one needs to give some accurate description  about the annealed free energy curve.
As it was remarked in \cite{BS08}, 
the annealed partition function $ \bbE^{Y, \rho} [Z_{t,\gb}^{Y,\p}]$ is simply the partition function of a homogeneous pinning model
$$\bbE^{Y,\rho}\left[Z_{t,\gb}^{Y,\p}\right]=\bbE^{X}\bbE^{Y,\rho}\left[e^{\gb L_t(X-Y,0)}
\ind_{\{(X-Y)_t=0\}}
\right].$$
Under $\bbE^{X}\bbE^{Y,\rho}$, $X-Y$ is a symmetric random walk with jump rate $(1+\rho)$. By rescaling time so that the random walk $X-Y$ has jump rate $1$, one obtains that 
\begin{equation}
 \tf^{\a}(\gb,\rho)= (1+\rho)\tf(\gb/(1+\rho),0).
\label{eq:Fann}
\end{equation}
We write $\tf(\gb)$ for $\tf(\gb,0)$.

The model is in fact {\sl exactly solvable} in the sense that one has an explicit formula for the free energy. This fact was remarked in the celebrated  paper of Fisher \cite{Fisher} for a discrete version of this model. We give a complete description of the pure model in the Appendix.
\smallskip

Let $p_t (\cdot) := \bbP^X(X_t=\cdot)$ denote the transition
probability kernel of $X$ at time $t$, and set $ G:=\int_0^{\infty}p_{t}(0)\dd t$ ($G<\infty$ when $d\geq 3$).

\begin{proposition}
\label{annealed}
For $d\ge 1$, the {\sl annealed} critical point is $\gb_c(0)=G^{-1}$ (we use the convention that $G^{-1}=0$ if $G=\infty$, for $d=1,2$),
and in view of \eqref{eq:Fann}, $\gb_c^{\a}(\rho)=(1+\rho)/G$.
One has also the critical behavior of the {\sl annealed} free energy:
\begin{itemize}
 \item for $d=1,3$,
\begin{equation}
\tf(\gb) \stackrel{\gb \downarrow \gb_c^{\a}}{\sim } c_0 (\gb-\gb_c^{\a})^{2}.
\end{equation}

\item for $d=2$,
\begin{equation}
\tf(\gb) \stackrel{\gb \downarrow \gb_c^{\a}}{=} \exp\left( -c_0\, \frac{1+o(1)}{\gb}  \right).
\end{equation}

\item for $d=4$,
\begin{equation}
\tf(\gb) \stackrel{\gb \downarrow \gb_c^{\a}}{\sim } c_0 (\gb-\gb_c^{\a})/\log (\gb-\gb_c^{\a}).
\end{equation}

\item for $d\geq 5$
\begin{equation}
\tf(\gb)  \stackrel{\gb \downarrow \gb_c^{\a}}{\sim } c_0 (\gb-\gb_c^{\a}).
\end{equation}
\end{itemize}
($c_0$ is a constant that can be made explicit, and that depends on $G$,
the dimension and the second moment of the jump kernel).\\
In any dimension, we also have
\begin{equation}\label{tralala}
 \lim_{\gb\to\infty}\tf(\gb)-\gb+1=0.
\end{equation}
\end{proposition}
Part of the above result (namely, the value of $\gb_c$), was proved in \cite{BS08}. We have included here also the asymptotic behavior near $\gb_c$ in order to know the specific heat exponent in any dimension. The knowledge of the annealed specific heat exponent (the free energy behaves like $(\gb-\gb_c)^{2-\alpha}$ when $\gb\to \gb_c^+$ where $\alpha$ is the specific heat exponent) allows to make prediction concerning disorder relevance.


\subsection{Harris criterion and disorder relevance}

 The physicist A.B. Harris gave a general criterion for disordered systems to predict disorder relevance
(for arbitrarily small strength of disorder) on a heuristic level. The criterion is based on the specific heat exponent of the pure system:
if the specific heat exponent is negative then disorder should be {\sl irrelevant}, 
if it is positive, disorder should be {\sl relevant}, and this corresponds to $d\ge 4$ for our model. The Harris criterion
gives no prediction for the marginal case when the specific heat exponent vanishes (and in that case, it is believed that disorder relevance depends on the model which is considered).

For the Random Walk Pinning Model, various pieces of work have brought this prediction on rigorous grounds \cite{BT, BS08, BS09}.
One of the main questions is to determine whether the annealed and quenched critical points differ or not.
If $\gb_c(\rho)=\gb_c^{\a}(\rho)$, then the disorder is said to be {\sl irrelevant}, and the quenched model's
critical behavior is believed to be similar to the one of the annealed model. Otherwise, the disorder shifts
the critical point ($\gb_c(\rho)>\gb_c^{\a}(\rho)$), and is said to be {\sl relevant}. The question of the relevance or
irrelevance of disorder for the RWPM is now solved, also for the marginal case $d=3$, both for the continuous time model 
\cite{BS08, BS09} and
for the discrete time model \cite{BT,BS08}.

\begin{theorem}[\cite{BS08,BS09}, Continuous time RWPM]
\label{shiftRWPM}
In dimension $d=1$ and $d=2$, one has ${\gb}_c(\rho)={\gb}_c^{\a}(\rho)=0$ for any positive $\rho$.
In dimension $d\geq 3$, one has $\gb_c>{\gb}_c^{\a}>0$ for each
$\rho>0$.
Moreover, we have a bound on the shift of the critical point :
\begin{itemize}
 \item For $d\geq 5$, there exists
  $a>0$ such that ${\gb}_c-{\gb}_c^{\a} \geq a \rho$ for
  all $\rho\in[0,1]$.
 \item For $d=4$ and for each $\gd>0$, there exists
  $a_{\gd}>0$ such that ${\gb}_c-{\gb}_c^{\a}\geq a_{\gd}
  \rho^{1+\gd}$ for all $\rho\in[0,1]$.
 \item For $d=3$ and for any $\gz>2$, there exists $c(\gz)>0$ such that
${\gb}_c-{\gb}_c^{\a}\geq e^{-c(\gz)\rho^{-\gz}}$ for all $\rho\in(0,1]$.
\end{itemize}
\end{theorem}

Let us also mention that the picture of disorder relevance/irrelevance
for the renewal pinning model (see \cite{Book} for a complete introduction to this model) is mostly complete, thanks to a series
of recent articles \cite{A06, DGLT,GLT}. It has been
showed that the Harris criterion is verified, and that in the marginal case, disorder is relevant.

Another issue that has been given much attention is the so called {\sl smoothing} of the free energy curve.
It is believed that for many systems, the presence of disorder makes the free energy curve more regular: the phase transition
is at least of second order (there is no discontinuity in the derivative). In particular this means that if
the annealed specific heat exponent is negative, quenched and annealed exponent have to differ.
This underlines disorder relevance, and gives further justification for the Harris criterion.

Smoothing type results have been shown by Aizenman and Wehr for disordered Ising model \cite{AW}, and more recently by Giacomin and Toninelli 
for the random pinning model based on renewal process \cite{GT05}  
(and also for a hierarchical version of the same model \cite{LT}). 


We also mention that there exist some peculiar pinning models for which
there is no smoothing phenomenon and the quenched and annealed systems
have always the same behavior, even if the critical points are different (see e.g. \cite{A06bis}).




 \subsection{Smoothing of the phase transition}

The first result we present for the disordered model is the smoothing of the free-energy curve around the phase transition.
This phenomenon occurs in dimension $d\geq3$. For $d=1,2$, the model is a bit different because of recurrence 
of the random walk in these dimensions, see later.

\begin{theorem}\label{smooth}
 For all $d\ge 3$, $\rho>0$, $\gb>0$, we have
\begin{equation}
 \tf(\gb,\rho)\le \frac{3d G^2}{\rho}(\gb-\gb_c(\rho))^2_+.
\end{equation}
\end{theorem}

This shows that if $d\geq 4$, the disorder makes the phase transition at least of second order,
whereas it is of first order for the {\sl annealed} model (see Proposition~\ref{annealed}).
The methods that has been used to prove the previous smoothing results \cite{AW,GT05} have been a strong source of inspiration 
for our proof, but, as the nature of the disorder is very different here, some new ideas are necessary.
A crucial point is to use an estimate on how $\tf(\gb,\rho)$ varies with $\rho$, which is present in \cite{BS09}.
It has been proved for the renewal pinning model that the critical exponent for the free-energy 
is related to the asymptotics of the number of contacts at the critical point \cite[Prop.\ 1.3]{L}. For the random walk pinning model
an analogous relation holds (where the number of contact is replaced by $L_t(X,Y)$) and gives the following result.
We include also its proof, which is very similar to what is done in \cite{L}, for the sake of completeness.

\begin{cor}
Let us fix $\rho>0$, $d\ge 3$ and $\gep>0$. Then, under $\bbP^{Y,\rho}$,
\begin{equation}
\lim_{t\to\infty}\mu^{Y,\p}_{t,\gb_c(\rho)}\left(L_t(X,Y)\ge t^{1/2+\gep}\right)= 0,
\label{cor:contact}
\end{equation}
in probability.
\end{cor}

\begin{rem}\rm
This result contrasts with what happens for the pure model ($\rho=0$),
where typically $L_t(X,0)\asymp t$ at $\gb_c$ for $d\ge 5$ (as shown in Corollary \ref{cor:purecontact}).
In analogy with what happens for the discrete renewal pinning model see \cite{Book},
one believes that at the critical temperature, $L_t(X,0)\ge t^{1-\gep}$ with high probability for any $\gep>0$ in dimension $d=4$.
This underlines a change in the critical behavior also in this dimension.
\end{rem}

\begin{proof}
 Suppose there exists some $c>0$ such that one can find an arbitrarily large value of $t$ for which
\begin{equation}\label{hy}
 \bbP^Y\left\{\mu^{Y,\p}_{t,\gb_c(\rho)}\left(L_t(X,Y)\ge t^{1/2+\gep}\right)\ge c\right\}\ge c.
\end{equation}
Then we define $t_0$ large enough such that the above holds, $u:=t_0^{-1/2}$ and $\gb_c:=\gb_c(\rho)$.
One has
\begin{equation}
 Z_{t_0,\gb_c+u}^{Y,\p} = \bbE^X \left[ e^{(\gb_c+u) L_{t_0}(X,Y)} \ind_{\{X_{t_0}=Y_{t_0}\}}\right] = 
        Z_{t_0,\gb_c}^{Y,\p}\ \mu_{t_0,\gb_c}^{Y,\p}\left(e^{u L_{t_0}(X,Y)}\right),
\end{equation}
so that
\begin{eqnarray}
\lefteqn{\bbE^{Y,\rho} \left[ \log Z_{t_0,\gb_c+u}^{Y,\p}\right]= \bbE^{Y,\rho} \left[ \log Z_{t_0,\gb_c}^{Y,\p}
+ \log \mu_{t_0,\gb_c}^{Y,\p}\Big(e^{u L_{t_0}(X,Y)}\Big)\right] } \nonumber\\
 & \ge &  \bbE^{Y,\rho}\left[\log p_{t_0}(Y_{t_0})\right]+ 
\bbE^{Y,\rho}\left[\log \left(ce^{u t_0^{1/2+\gep}}\right)
\ind_{\big\{\mu^{Y,\p}_{t_0,\gb_c}\left(L_{t_0}(X,Y)\ge {t_0}^{1/2+\gep}\right)\ge c\big\}}\right] \nonumber\\
  & \ge &
-d\log t_0+ c(t_0^{\gep}+\log c)\ge t_0^{\gep/2},
\end{eqnarray}
where in the first inequality we used that $Z_{t_0,\gb_c}^{Y,\p}\geq p_{t_0}(Y_{t_0})$ (recall the notation introduced just before
Proposition~\ref{annealed}, this is just using the fact that $\gb_c\ge 0$).
The second inequality uses an estimate from \cite[Lemma 3.1]{BS08} which is valid if $t_0$ is large enough for the first term,
and \eqref{hy} for the second term. The last inequality is valid
if $t_0$ is large enough.
This implies, by \eqref{superaddi}
\begin{equation}
\tf\big(\gb_c(\rho)+u,\rho\big)\ge \frac{1}{t_0}\bbE^{Y,\rho} \left[\log Z_{t_0,\gb_c+u}^{Y,\p}\right]\ge t_0^{\gep/2-1}\ge u^{2-\gep}.
\end{equation}
This contradicts Theorem \ref{smooth}, therefore \eqref{hy} cannot hold.
\end{proof}

\smallskip

In dimension $1$ or $2$, the situation is a bit different due to recurrence of the random walk.
In dimension $d=2$, the coincidence of quenched and annealed critical point,
and the fact that the phase transition is of infinite order of the annealed system implies
that the phase transition is also of infinite order (i.e.\ smoother than any power of $(\gb-\gb_c)$\ ) for the quenched system.
In dimension $d=1$, one also shows that disorder does not change the nature of the phase transition  (or at least not in a significant way).

\begin{proposition}[Quenched free-energy at high temperature for $d=1$]
There exist a constant $c>0$ such that for any $\rho$ there exists $\gb_0$ such that
\begin{equation}
 \tf(\gb,\rho)\geq \frac{c}{1+\rho} \gb^{2} \log (1/\gb) ^{-1}, \quad \forall \gb \in [0,\gb_0].
\end{equation}
Thus, $\tf$ and $\tf^{\a}$ have the same critical exponent.
\end{proposition}

\begin{proof}

By Jensen inequality one has (for some constant $C_1$) that for every $t$ and $\rho\geq 0$
\begin{equation}
 \bbE^{Y,\rho}\left[ \log \bbE^X \left[ e^{\gb L_t(X,Y)}\  \big| \  X_t=Y_t \right] \right] \\
\geq \gb \bbE^{Y,\rho} \bbE^X \left[ L_t(X,Y)\  \big| \ X_t=Y_t \right]\ge C_1 \frac{\gb \sqrt{t}}{\sqrt{1+\rho}}.
\end{equation}
The last inequality can be obtain by integrating the local central limit Theorem (see \cite[Prop. 7.9, Ch. II]{Spitz} for the
discrete time version, the proof being identical for continuous time). Therefore
\begin{eqnarray}
 \bbE^{Y,\rho} \left[\log Z^{ Y,\p}_{t,\gb}\right] &=& \bbE^{Y,\rho}\left[ \log \bbE^X \left[ e^{\gb L_t(X,Y)} \big| X_t=Y_t \right] \right]
                 +\bbE^{Y,\rho}\left[ \log \bbP^X \left( X_t=Y_t \right)\right] \nonumber \\
    &\geq & C_1 \frac{\gb \sqrt{t}}{\sqrt{1+\rho}} - \log t,
\end{eqnarray}
where we also used \cite[Lemma 3.1]{BS08} to bound the second term (the bound being valid for $t$ large enough, say $t\geq t_0(\rho)$).
Now, if we set $T:= C_2(1+\rho)\gb^{-2} [\log (1/\gb)]^2$, the previous inequality holds for all $\gb\leq t_0^{-1/2}$, and gives
\begin{equation}
\bbE^{Y,\rho}\left[\log  Z^{ Y,\p}_{T,\gb} \right] \geq C_1\sqrt{C_2}\log (1/\gb) + 2\log \gb + O\big( \log \log (1/\gb) \big)     \ge  \log (1/\gb)
\end{equation}
if $C_2$ is large enough.
From \eqref{sup}, we finally have
\begin{equation}
\tf(\gb)  \geq  \frac{1}{T}\bbE^{Y,\rho}\left[\log  Z^{ Y,\p}_{T,\gb} \right] \ge \frac{1}{C_2(1+\rho)}\gb^2 \log (1/\gb)^{-1}.
\end{equation}
\end{proof}

\begin{rem}\rm
We believe that the factor $\log (1/\gb) ^{-1}$ above is an artifact of the proof and that $\tf(\gb,\rho)\sim c(\rho)\gb^2$.
A clear reason to believe so is to consider an alternative Brownian
model where $(Y_t)_{t\ge 0}$ is a realization Brownian motion with
covariance function $\bE^Y [Y_s, Y_t]=\rho(s\wedge t)$. The partition function is given by
\begin{equation}
 \mathcal Z^Y_{t,\gb}=\bE^X\left[e^{\gb L_t(X,Y)}\right],
\end{equation}
where under $\bP^X$, $X$ is a standard Brownian Motion (independent of $Y$) and $L_t(X,Y)$ the intersection local time between $X$ and $Y$.
For this model, Brownian scaling implies that $L_{t\gb^2}(X,Y)\stackrel{(\mathrm{law})}{=} \gb L_{t}(X,Y)$, which implies that there exists a constant $c(\rho)$ ($=\bar \tf(1,\rho)$) such that for all $\gb \ge 0$,
\begin{equation}
 \bar \tf(\gb,\rho):=\lim_{t\to\infty} \frac{1}{t}\bE^Y \log \mathcal Z^Y_{t,\gb}=c(\rho) \gb^2.
\end{equation}
This model should be the high-temperature scaling-limit of our random walk pinning model and hence share the same critical properties.

\end{rem}

\subsection{Low temperature asymptotics}

The quenched low-temperature asymptotic also exhibits contrasts with the annealed one. The reason is that to optimize the local-time,
$X$ has to follow $Y$ closely, which has an extra entropic cost. In the annealed case, one can force $X$ not to jump.
For the sake of simplicity, we present the result only in the case of the simple symmetric random walk in $\Z^d$ (for any 
$d\ge 1$) but the result holds in the more general framework given in Section \ref{fw}. This result gives again a contrasts with the pure model, see \eqref{tralala}.

\begin{theorem} \label{lowtemp}
When $Y$ is the simple symmetric random walk in $\Z^d$,
one has
\begin{equation}
 \tf(\gb,\rho)=\gb-\rho\log d\gb-1+o(1)  \quad \text{ as } \gb \to \infty.
\end{equation}
\end{theorem}

In general, for a walk $Y$ with a kernel jump $p_Y$ which as finite second moment, the result also holds with $\log d$ replaced by
$-\sum_{x\in \Z^d} p_Y(x)\log(p_Y(x)/2)$.

\begin{rem}\rm The proof of Theorem \ref{lowtemp} does not only gives the result but also a clear idea of how a typical path $X$ behaves
under the polymer measure at high temperature.
Essentially $X$ follows every jump of $Y$, and the distance between jumps of $X$ and $Y$ are i.i.d.\ exponential variables of mean $1/\gb$.
In particular the asymptotic contact fraction is close to 
$1-\rho \gb^{-1}$ (whereas it is of order $1-O(\gb^{-2})$ for the pure model).
\end{rem}

\medskip

The sequel of the paper is organized as follows:
\begin{itemize}
\item In Section \ref{proof} we prove Theorem \ref{smooth},
\item In Section \ref{lowtempp} we prove Theorem \ref{lowtemp},
\item In the Appendix  we prove some statements for the pure model, including Proposition \ref{annealed}.
\end{itemize}
Section \ref{proof} and Section \ref{lowtempp} are independent.

\section{Proof of Theorem \ref{smooth}}

\label{proof}
In the proof, we make use of the following three statements. The first two are extracted from
 Proposition~\ref{annealed} and Theorem~\ref{shiftRWPM}, the third one is extracted from \cite{BS09}.

\begin{proposition}\label{monont}
 For $d\ge 3$, we have
\begin{itemize}
 \item[(i)] $\gb_c^{\a}(\rho)= \frac{1+\rho}{G}$,
 \item[(ii)] for any $\rho>0$, one has $\gb_c^{\a}(\rho)<\gb_c(\rho)$,
 \item[(iii)] the function $\rho\mapsto \gb_c(\rho)/(1+\rho)$ is non-decreasing \cite[Thm 1.3]{BS09}.
\end{itemize}
\end{proposition}
Let $\rho$ be fixed and $d\ge 3$ be fixed.
Given $\gb>\gb_c(\rho)$, we define $\rho'=\rho'(\gb)$ by
\begin{equation}
 \rho':=\rho+G(\gb-\gb_c(\rho)).
\end{equation}
Note that $\tf(\rho',\gb)=0$. Indeed, we have
\begin{equation}
 \frac{1+\rho'}{1+\rho}=1+G\frac{\gb-\gb_c(\rho)}{1+\rho}=1+\frac{\gb-\gb_c(\rho)}{\gb_c^{\a}(\rho)}\ge \frac{\gb}{\gb_{c}(\rho)},
\end{equation}
so that by $(iii)$ of the above proposition, $\gb\le \gb_c(\rho')$.

Our strategy to prove Theorem \ref{smooth} is to find a lower bound for $\tf(\rho',\gb)$ that involves $\tf(\rho,\gb)$, by considering
the contribution of exceptional (under $\bbP^{Y,\rho'}$) stretches where the empirical jump rate of $Y$ is of order $\rho$.

First we bound from below the probability that under $\bbP^{Y,\rho'}$, the partition function $ Z^{Y,\p}_{L,\gb}$ 
is greater than $\exp\big(L(1-\gep)\tf(\gb,\rho)\big)$.

\begin{lemma}\label{boundp}
For any $\gep>0$, one can find $L_0$ (depending on $\gb, \rho$ and $\gep$) such that for all $L\geq L_0$
\begin{equation}
 \log \left(\bbP^{Y,\rho'}\left\{\log Z^{Y,\p}_{L,\gb} > L(1-\gep)\tf(\gb,\rho) \right\}\right)\ge -L\frac{(\rho'-\rho)^2}{\rho'}-\log 4.
\end{equation}
\end{lemma}

\begin{proof}

From the definition of the free-energy one can find $L_0$ such that for all $L\geq L_0$
\begin{equation}
\bbP^{Y,\rho}(A):= \bbP^{Y,\rho}\left\{\log Z_{L,\gb}^{Y,\p} > L(1-\gep)\tf(\gb,\rho) \right\} \ge 1/2.
\end{equation}
By Cauchy-Schwartz inequality, we have
\begin{equation}\label{CStz}
 1/4\le  \bbP^{Y,\rho}(A)^2 =\left(\bbE^{Y,\rho'}\left[\frac{\dd \bbP^{Y,\rho}}{\dd \bbP^{Y,\rho'}}\ind_A\right]\right)^2
\le \bbE^{Y,\rho'}\left[\left(\frac{\dd \bbP^{Y,\rho}}{\dd \bbP^{Y,\rho'}}\right)^2\right]
\bbP^{Y,\rho'}(A).
\end{equation}
Let $\gk_L^Y$ denote he number of jumps of the walk $Y_t$ in $[0,L]$. Under $\bbP^{Y,\rho}$, it is a Poisson 
variable of mean $\rho L$. One has
\begin{equation}
 \frac{\dd \bbP^{Y,\rho}}{\dd \bbP^{Y,\rho'}}=e^{L(\rho'-\rho)}\left(\frac{\rho L}{\rho' L}\right)^{\gk^Y_L},
\end{equation}
and therefore
\begin{equation}
\bbE^{Y,\rho'}\left[\left(\frac{\dd \bbP^{Y,\rho}}{\dd \bbP^{Y,\rho'}}\right)^2\right]
=e^{L(\rho'-2\rho)}\sum_{k=0}^{\infty}\frac{1}{k!}\left(\frac{(\rho L)^2}{\rho' L }\right)^k
=e^{L\frac{(\rho'-\rho)^2}{\rho'}},
\end{equation}
which inserted in \eqref{CStz} gives the result.
\end{proof}

We keep the notation $A:= \left\{\log Z^{Y,\p}_{L,\gb} > L(1-\gep)\tf(\gb,\rho) \right\}$ and write $q:=\bbP^{Y,\rho'}(A)$. 
Let an arbitrary $\gep>0$ be fixed,
consider $L$ large enough so that Lemma \ref{boundp} is valid, and that some conditions later mentioned in the proof are fullfilled.


For our purpose we find a lower bound involving $\tf(\gb,\rho)$ for  $\bbE^{Y,\rho'}\left[\log Z_{T,\gb}^{Y,\p}\right]$ for a system length  
$T:= L\lceil q^{-1} \rceil$. Then we use the fact that $\bbE^{Y,\rho'}\left[\log Z_{T,\gb}^{Y,\p}\right]\le T\tf(\gb,\rho')=0$ (recall \eqref{sup}).

We divide the length of the system into $\lceil q^{-1} \rceil$ blocks of size $L$,
${B_i:= [(i-1)L, iL]}$, for $i\in \{1,\dots, \lceil q^{-1}\rceil\}$.
With this definition, under $\bbP^{Y,\rho'}$, the random variables
 $(Z^{Y,\p}_{B_i,\gb})$ are i.i.d. distributed with the same distribution as $Z_{L,\gb}^{Y,\p}$.




Define
\begin{equation}
H := \left\{ \exists  i \in [1, \lceil q^{-1} \rceil]\cap \bbZ, \
               \log Z^{Y,\p}_{B_i,\gb}\ge L(1-\gep)\tf(\gb,\rho)  \right\}.
\end{equation}
We define also $\cA=\cA(Y)$ the set of blocks $B_i$ such that $\log Z^{Y,\p}_{B_i,\gb}\ge  L(1-\gep)\tf(\gb,\rho)$.

Note that with our choice for the number of blocks considered, $\bbP^{Y,\rho'}(H)=1-(1-q)^{\lceil q^{-1} \rceil}$ is 
uniformly bounded away from 
zero and one.


One splits $\bbE^{Y,\rho'}\left[\log Z_{T,\gb}^{Y,\p}\right]$ into two contributions according to the occurrence of $H$.
\begin{equation}\label{toutit}
\bbE^{Y,\rho'}\left[\log Z_{T,\gb}^{Y,\p}\right] = \Ey\left[\ind_{H^c}\log Z_{T,\gb}^{Y,\p} \right] + \Ey\left[\ind_{H}\log Z_{T,\gb}^{Y,\p} \right].
\end{equation}

$\bullet$ The first term is dealt with easily, using
that $Z_{T,\gb}^{Y,\p} \geq p_T(Y_{T})$, and then
\begin{equation}
\Ey\left[\ind_{H^c}\log Z_{T,\gb}^{Y,\p} \right] \geq \Ey \left[ \log p_T (Y_{T})\right]\label{fterm}
   \geq  -(1+\gep)\frac{d}{2} \log T,
\end{equation}
where the last estimate comes from \cite[Lemma 3.1]{BS08}, (provided that $T=L\lceil q^{-1} \rceil$ is large enough).

$\bullet$ For the second term, we only have to decompose
the expectation according to the position of
the first block for which $ \log Z_{B_i,\gb}^{Y,\p}\ge L(1-\gep)\tf(\gb,\rho)$, 


\begin{equation}
 \Ey\left[\ind_{H}\log Z_{T,\gb}^{Y,\p} \right]  \geq 
    \sum_{i=1}^{\lceil q^{-1} \rceil} \Ey \left[ \ind_{\{{B_i}\in\cA, B_j\notin \cA \  \forall 1\le j<i \}} \log Z_{T,\gb}^{Y,\p} \right].
\label{eq:Gblock}
\end{equation}
By \eqref{superaddi}, one obtains on the event 
$\{{B_i}\in\cA, B_j\notin \cA \ \ \forall 1\le j<i \}$, that
\begin{multline}
 Z_{T,\gb}^{Y,\p}\ge Z^{Y,\p}_{(i-1)L,\gb}Z_{B_i,\gb}^{Y,\p} Z^{Y,\p}_{[T-iL,T],\gb} \\
  \ge e^{L(1-\gep)\tf(\gb,\rho)}p_{(i-1)L}(Y_{(i-1)L})p_{T-iL}(Y_T-Y_{iL}).
\end{multline}
Therefore we get
\begin{multline}
 \Ey \left[ \ind_{\{{B_i}\in\cA, B_j\notin \cA \ \forall 1\le j<i \}} \log Z_{T,\gb}^{Y,\p} \right]\ge\\
 \Ey\Big[ \ind_{\{{B_i}\in\cA, B_j\notin \cA \ \forall 1\le j<i \}} \big( L(1-\gep)\tf(\gb,\rho)\\
+ \log p_{(i-1)L}(Y_{(i-1)L})+\log p_{T-iL}(Y_T-Y_{iL})\big) \Big].
\end{multline}
We can estimate separately the three terms on the right-hand side. Note that by block independence,
$\bbP^{Y,\rho'}\left[ B_i\in\cA, B_j\notin \cA \ \forall 1\le j<i \right]=q(1-q)^{(i-1)}$. This gives the value of the first term.
When $i\ne \lceil q^{-1} \rceil$ (in which case the third term is zero), using block
independence and translation invariance, the third one can be estimated 
as follows 
\begin{multline}
 \Ey\left[ \ind_{\{{B_i}\in\cA, B_j\notin \cA\ \forall 1\le j<i \}}\log p_{T-iL}(Y_T-Y_{iL})\right]=
q(1-q)^{i-1}\Ey\left[p_{T-iL}(Y_{T-iL})\right]\\
\ge -(1-q)^{i-1} q(1+\gep)\frac{d}{2} \log(T-iL)\ge -q(1+\gep)\frac{d}{2} \log T,
\end{multline}
where the last inequality is given again by \cite[Lemma 3.1]{BS08}, provided $L$ is large enough.
When $i\ne 1$ (in which case the second term is zero), block independence gives us 
\begin{eqnarray}
\lefteqn{ \Ey\left[ \ind_{\{{B_i}\in\cA, B_j\notin \cA \ \forall 1\le j<i \}}\log p_{(i-1)L}(Y_{(i-1)L})\right]}  \nonumber\\
 & &  =  q\Ey\left[ \ind_{\{B_j\notin \cA \ \forall 1\le j<i \}}\log p_{(i-1)L}(Y_{(i-1)L})\right]
\ge q \Ey\left[\log p_{(i-1)L}(Y_{(i-1)L})\right] \nonumber\\
 & & \ge -q(1+\gep)\frac{d}{2}\log((i-1)L) \ge -q(1+\gep)\frac{d}{2} \log T.
\end{eqnarray}
Summing along all the contributions, one gets
\begin{equation}
\Ey\left[\ind_{H}Z_{T,\gb}^{Y,\p} \right] \ge\left(1-(1-q)^{\lceil q^{-1}\rceil}\right) L(1-\gep)\tf(\rho,\gb)-(1+\gep)d\log T.
\end{equation}
Together with \eqref{toutit} and \eqref{fterm} this gives
\begin{equation}
0\ge \Ey\left[\log Z_{T,\gb}^{Y,\p} \right] \ge\left(1-e^{-1}\right) L(1-\gep)\tf(\rho,\gb)-(1+\gep)\frac{3d}{2}\log T,
\end{equation}
and hence
\begin{equation}
 \tf(\rho,\gb)\le \frac{1+\gep}{1-\gep}\frac{3d}{2(1-e^{-1})}\frac{\log ( L\lceil q^{-1}\rceil)}{L}.
\end{equation}
From Lemma \ref{boundp}, one has (when $L$ is large enough)
\begin{equation}
 \frac{\log (L\lceil q^{-1}\rceil)}{L}\le (1+\gep)\frac{(\rho'-\rho)^2}{\rho'}\le (1+\gep)\frac{G^2\left(\gb-\gb_c(\rho)\right)^2}{\rho},
\end{equation}
which, as $\gep$ is arbitrary, gives the result (here we also use $(1-e^{-1})> 1/2)$.

\qed

\section{Proof of Theorem \ref{lowtemp}}\label{lowtempp}

Our bounds are obtained by decomposing the partition function into a product, each term of the product corresponding to a time interval.
 
To describe our decomposition, we need some definitions.
We fix a typical realization of $Y$. Let $T_i$ be the time of the $i$-th jump.
For some $\gb$ (large) fixed and $i\geq 1$, we define the times
\begin{equation}
 \begin{split}
  T_i^- = T_i - \gep_i^- \indent \text{with } \gep_i^- =\gb^{-2/3} \wedge \frac{T_i-T_{i-1}}{2}, \\
  T_i^+ = T_i + \gep_i^+ \indent \text{with } \gep_i^+ =\gb^{-2/3} \wedge \frac{T_{i+1}-T_i}{2},
 \end{split}
\end{equation}
where we used the convention that $T_0=0$, and set also $T_0^+=0$. The value $\gb^{-2/3}$ in the definition is an {\it ad hoc}
choice for the proof of the upper bound, and has no deep signification.
We also define $\gep_i:=\gep_i^- + \gep_i^+=T_i^+ - T_i^-$.
We bound $Z^{Y,\p}_{T_k^+,\gb}$ by bounding the contributions of the intervals $[T_{i-1}^+, T_{i}^-)$ and $[T_i^-, T_i^+)$, $i=1 \dots k$.

\textbf{Lower Bound.}
To lower bound
$\log Z^{Y,\p}_{T_k^+,\gb}$, we use superadditivity:
\begin{equation}
\label{eq:lower}
 \log Z^{Y,\p}_{T_k^+,\gb} \geq \sum_{i=1}^k \left(\log Z^{Y,\p}_{[T_{i-1}^+, T_i^-]} + \log Z^{Y,\p}_{[T_{i}^-, T_i^+]}\right).
\end{equation}
Let us note that $Y$ makes no jump on $[T_{i-1}^+, T_i^-]$,
so that
\begin{equation}
Z^{Y,\p}_{[T_{i-1}^+, T_i^-]}= \bbE^X\left[e^{L_{T_i^- - T_{i-1}^+}(X,0)}\ind_{\{X_{T_i^- - T_{i-1}^+}=0\}}\right].
\end{equation}
Then, constraining $X$ not to jump either, we get that for all $t\geq 0$
\begin{equation}
 \bbE^X\left[e^{L_t(X,0)}\ind_{\{X_t=0\}}\right]\ge e^{\gb t} \bbP^X \left(X_s=0 \text{ for all } 0\le s\le t \right) = e^{(\gb-1)t}.
\label{eq:low1}
\end{equation}
To bound $\log Z^{Y,\p}_{[T_{i}^-, T_i^+]}$, let use notice that for $i\geq1$, $Y$ makes one jump (and only one)
in the interval $[T_i^-, T_i^+)$ (of length $\gep_i$). Hence
\begin{equation} 
Z^{Y,\p}_{[T_{i}^-, T_i^+]}=\bbE^X_{0}\left[e^{L_{\gep_i}(X,Y^{(i)})} \ind_{\{X_{\gep_i}=Y_{\gep_i}\}} \right],
 \end{equation}
where $Y^{(i)}=(Y^{(i)}_s)_{s\in [0,\gep_i)}$ is defined by $Y^{(i)}_s=0$ for
$s\in[0,\gep_i^-)$,  $Y^{(i)}_s=e_1=(1,0,\dots,0)$ for $s \in [\gep_i^-,\gep_i)$ (by symmetry and the fact that the random walk is neirest neighbor the direction of the jump has no importance).
We will compute the contribution of the terms in which
$X$ does one and only one jump, furthermore in the right direction $e_1$.
We have
\begin{multline}
\label{eq:low2}
\bbE^X_{0}\left[e^{L_{\gep_i}(X,Y^{(i)})}\ind_{\{ X \text{ makes one jump, } X_{\gep_i}=e_1\}}  \right]=
\frac{1}{2d}e^{-\gep_i}\int_{0}^{\gep_i} e^{\gb(\gep_i-|s-\gep_i^-|)}\dd s\\
=\frac{1}{2d}e^{-\gep_i}\int_{-\gep_i^-}^{\gep_i^+} e^{\gb(\gep_i-|s'|)}\dd s'=
\frac{e^{(\gb-1)\gep_i}}{d\gb }\left[1-\frac{e^{-\gb \gep_i^-}}{2}-\frac{e^{-\gb \gep_i^+}}{2} \right].
\end{multline}
The term $\frac{1}{2d}e^{-\gep_i} \dd s$ is the probability of having only one jump in 
$[0,\gep_i]$ located in the time increment $[s,s+\dd s]$, that goes in the right direction (cf.\ factor $(2d)^{-1}$), $\gep_i-|s-\gep_i^-|$ is the value of the intersection
time of $X$ and $Y$ on $[0,\gep_i)$ if $X$ jumps at time $s$.
Combining \eqref{eq:low1}-\eqref{eq:low2} with the inequality \eqref{eq:lower}, we obtain
\begin{equation}
\frac{1}{T_k^+}\log Z^{Y,\p}_{T_k^+,\gb}\ge  (\gb-1)+
   \frac{k}{T_k^+}\left[-\log(d\gb)+ \frac{1}{k}\sum_{i=1}^{k}\log \left(1-\frac{e^{-\gb\gep_i^-}+e^{-\gb\gep_i^+}}{2}\right)\right].
\end{equation}

The sequence $\log \left(1-\frac12 \left( e^{-\gb\gep_i^-}+e^{-\gb\gep_i^+} \right)\right)$ is ergodic (the dependence between terms has range only one).
Then using the ergodic theorem one obtains that $\bbP^{Y,\rho}$-a.s.
\begin{equation}
 \lim_{k\to\infty} \frac{1}{k}\sum_{i=1}^{k}\log \left(1-\frac{e^{-\gb\gep_i^-}+e^{-\gb\gep_i^+}}{2}\right)
=\bbE^{Y,\rho}\left[\log \left(1-\frac{e^{-\gb\gep_1^-}+e^{-\gb\gep_1^+}}{2} \right)\right]=o(1),
\end{equation}
where $o(1)$ is with respect to $\gb\to \infty$. The last inequality is easy to get, as $\gep_1^{\pm}$ are
truncated exponential variables of mean $1/2$ and that the truncation at $\gb^{-1/3}$ is harmless.
Moreover, by the law of large numbers, we have that $\bbP^{Y,\rho}$-a.s,
\begin{equation}
 \lim_{k\to \infty} \frac{k}{T_k^+}= \lim_{k\to \infty} \frac{k}{T_k}= \rho.
\end{equation}
This gives us
\begin{equation}
 \lim_{k\to\infty} \frac{1}{T_k^+}\log Z^{Y,\p}_{T_k^+,\gb} \ge (\gb-1)-\rho\log d\gb+o(1).
\end{equation}
\smallskip

\textbf{Upper Bound.}
We are now ready to prove the upper bound.
We cut the trajectory $X$ on the intervals $[T_{i-1}^+,T_i^-)$ and $[T_i^-, T_i^+)$ for $i\geq 1$ and use the properties
of $Y$ on these intervals, the way we did for the lower bound. In order to get an upper bound,
we have to maximize over the contribution of intermediate points, 
\begin{equation}
 Z^{Y,\p}_{T_k^+,\gb}\le  \prod_{i=1}^{k} \max_{x_1\in \Z^d} \bbE_{x_1}^X\left[e^{\gb L_{T_i^- - T_{i-1}^+}(X,0)}\right]
   \max_{x_2\in \Z^d} \bbE_{x_2}^X\left[e^{\gb L_{\gep_i}(X,Y^{(i)})}\right].
\end{equation}
We can bound the first part of the terms by using Lemma \ref{cocorico} (that we prove later on):
\begin{equation}
 \max_{x\in \Z^d} \bbE_{x}^X\left[e^{\gb L_{t}(0,X)}\right]=
   \bbE_{0}^X\left[e^{L_{t}(0,X)}\right]\le e^{\left(\gb-1+\frac{1}{\gb}\right)t}\left(1+\frac{1}{\gb}\right),
\end{equation}
where the first equality is due to Markov property for $X$ applied at the
first hitting time of zero, and the fact that $\bbE_{0}^X\left[e^{L_{t}(0,X)}\right]$ is a non-decreasing function of $t$.

For the other terms one has to analyze the contributions of all possible trajectories of $X$.
The main contribution is given by paths $X$ starting from zero that make one jump and such that $X_{\gep_i}=e_1$: we already computed
the value of this contribution in \eqref{eq:low2}.
If $X$ makes no jump or one jump but not in the right direction (or if $X$ makes at most one jump but does not start from zero),
it spends some portion of the time away from $Y$ and then $L_{\gep_i}(X,Y^{(i)})\le \gep_1^- \vee \gep_i^+ \leq \gb^{-2/3}$.
Therefore the total contribution of such paths is bounded by
$e^{\gb^{1/3}}$.
The probability that $X$ makes more than two jumps is bounded by $4 \gb^{-4/3}$ if $\gb$ is large enough (the number of jump
is a Poisson variable of parameter $\gep_i$, which is at most $2\gb^{-2/3}$).
In addition $e^{\gb L_{\gep_i}(X,Y^{(0)})}\le e^{\gb \gep_i}$ so that the total contribution of paths making more than two jumps is bounded by 
$4 \gb^{-4/3}  e^{\gb \gep_i}$.
Hence we have
\begin{eqnarray}
  \max_{x\in \Z^d} \bbE_x^X\left[e^{\gb L_{\gep_i}(X,Y^{(i)})}\right] &\le&
         \frac{1}{d\gb}e^{(\gb-1)\gep_i} + e^{\gb^{1/3}} + 4\gb^{-4/3}  e^{\gb \gep_i} \nonumber\\
     &\le& \frac{1}{d\gb}e^{(\gb-1)\gep_i} \left( 1+ C(\gb e^{\gb^{1/3}-\gb\gep_i}+\gb^{-1/3}) \right).
\end{eqnarray}
Combining all these inequalities one finally gets
\begin{equation}
  Z^{Y,\p}_{T_k^+,\gb}\le e^{(\gb-1+\gb^{-1})T_k^+}(1+\gb^{-1})^k(d\gb)^{-k}\prod_{i=1}^k \left( 1+  C (\gb e^{\gb^{1/3}-\gb\gep_i}+\gb^{-1/3}) \right),
\end{equation}
and hence
\begin{multline}
 \frac{1}{T_k ^+} \log Z^{Y,\p}_{T_k^+,\gb} \le \gb-1+\gb^{-1} \\
     +  \frac{k}{T_k^+} \left[ \log (1+\gb^{-1}) -\log(d\gb) +
   \frac{1}{k}\sum_{i=1}^k \log \left( 1+  C (\gb e^{\gb^{1/3} - \gb \gep_i}+\gb^{-1/3}) \right) \right].
\end{multline}

Applying the ergodic theorem, one obtains that $\bbP^{Y,\rho}$-a.s.
\begin{eqnarray}
 \lim_{k\to\infty} \frac{1}{k}\sum_{i=1}^{k}\log \left( 1+  C (\gb e^{\gb^{1/3} - \gb \gep_i}+\gb^{-1/3})\right)
 & =& \bbE^{Y,\rho}\left[\log \left( 1+  C (\gb e^{\gb^{1/3} - \gb \gep_1}+\gb^{-1/3}) \right)\right] \nonumber\\
 &= &  o(1) ,
\end{eqnarray}
where we used the definition of $\gep_1$ to estimate the last expectation ($\gep_i$ is equal to $2\gb^{-2/3}$ with probability $1-O(\gb^{-2/3})$ when $\gb$ is large).

Furthermore, as already noticed, $k/T_k^+$ converges almost surely to $\rho$, so that we have
\begin{equation}
\lim_{k\to \infty} \frac{1}{T_k^+} \log Z^{Y,\p}_{T_k^+,\gb}\le \gb-1-\rho\log\gb d+o(1).
\end{equation}
\qed

\appendix

\section{The homogeneous case}
\label{sec:homo}

We give in this section several results on the pure model, which are to be compared with the results on the quenched system.
In the homogeneous case, (when $\rho=0$) the model is just the pinning of a random walk on a deterministic defect line $ \bbR_+\times \{0\}$. 
It turns out here that a more general view point makes the problem easier to solve, and that is the reason why we introduce now a more general version of our pinning model.

We consider two increasing sequences $(\tau'_i)_{i\geq1}$ and $(\tau_i)_{i\geq0}$ such that $\tau_0=0$, and
$(\tau'_{i+1}-\tau_i)_{i\ge 0}$ and $(\tau_i-\tau'_i)_{i\ge 1}$ are two independent i.i.d.\ sequences,
where $\tau'_1$ is a mean~$1$ exponential variable, and where the distribution of $\tau_1-\tau'_1$ has
support in $\R_+\cup \{\infty\}$ (if $\tau_n=\infty$ for some $n$, we choose by
convention $\tau'_k$ and $\tau_k=\infty$ for all $k\ge n$). We further assume that the distribution of $(\tau_{1}-\tau'_1)_{1\ge 0}$, when restricted to  $\bbR_+$, is
absolutely continuous with respect to the Lebesgue measure, with density
that we denote by $K(\cdot)$, and
that $\int_0^{\infty} \exp(\gep t)K(t)\dd t=\infty$ for all $\gep>0$.
We denote by
$\mu$ the joint law of the two sequences, and we remark that under $\mu$, both sequences $(\tau_i)_{i\geq 0}$ and $(\tau'_i-\tau'_1)$
are renewal sequences.
We may use the notation $K(\infty)=\mu(\tau_1-\tau'_1=\infty)=\mu(\tau_1=\infty)$.


Set $\cT := \bigcup_{i=0}^{\infty} [\tau_i,\tau'_i)$, and call it the set of contact.

Given $\gb\in \bbR$, we now modify the law of the sequences $(\tau'_i)_{i\geq 0}$ and $(\tau_i)_{i\geq 0}$ by introducing a Gibbs transform
$\mu^{\p}_{t,\gb}$
of the measure $\mu$
\begin{equation}
\label{eq:defmu2}
 \frac{\dd \mu_{t,\gb}^{\p}}{\dd \mu} = \frac{e^{ \gb |\cT \cap [0,t]|}}{Z_{t,\gb}^{\p}}  \ind_{\{t\in\cT\}},
\end{equation}
where $|A|$ stands for the Lebesgue measure of a set $A \subset\bbR$, and where
\begin{equation}
Z_{t,\gb}^{\p}=\mu\left[e^{\gb |\cT\cap [0,t]|} \ind_{\{t\in\cT\}}\right].
\end{equation}
We also define
\begin{equation}
 \tf(\gb):=\lim_{t\to\infty}\frac{1}{t}\log Z_{t,\gb}^{\p}
\end{equation}
which is well defined, by superadditivity.

\begin{rem}\rm
In the case of a continuous time random walk $X$ with jump rate $1$ (of law $\bbP^X$),
we set $\tau_0=0$ and for all $i\ge 1$
\begin{equation}\label{toto} 
\begin{split}
 \tau'_i &:=\inf\{t>\tau_{i-1}, X_t\ne 0\},\\
 \tau_i &:=\inf\{t>\tau'_i, X_t= 0\}.
\end{split}\end{equation}
One can check that $(\tau_i-\tau'_i)_{i\ge 1}$ and $(\tau'_{i+1}-\tau_i)_{i\ge 0}$ are
independent i.i.d.\ sequences (for $d\ge 3$, there are only finitely many terms in the sequences)
that satisfies the assumptions given above. Therefore, our definition \eqref{eq:defmu} of $\mu_{t,\gb}^{\p}$ (with $Y$ replaced with $0$)
coincides with the one of \eqref{eq:defmu2}. This underlines two things:
\begin{itemize}
\item  The pinning model we present in this section is indeed a generalization of the pure (or annealed) model
for the random walk-pinning.
\item  In annealed random-walk pinning, the Gibbs transformation changes only the return time to zero and the time $X$ spends on zero. Conditionally on these times,
the law of the excursions out of the origin remains the same that under $\bbP^X$.
\end{itemize}
 
\end{rem}


We can describe the measure $\mu_{t,\gb}^{\p}$ in a very simple way, because we are interested
only in the law of $\cT\cap[0,t]$ (as it is the only part that is modified by the Gibbs transformation).
We introduce some definitions to describe the measure.

If $(1-\gb)^{-1}\int_0^{\infty} K(t)\dd t\ge 1$ or $\gb\ge 1$, let $b\ge 0$ be defined by
\begin{equation}
 (1-\gb+b)^{-1} \int_0^{\infty} e^{-b t} K(t)\dd t=1
\label{defb}
\end{equation}
and $b=0$ if $(1-\gb)^{-1}\int_0^{\infty} K(t) \dd t< 1$.\\
For notational reasons, define $\gl:=(1-\gb+b)$.
Then, we define $\tilde K^{\gb} (t):= \gl^{-1} e^{-b t} K(t)$ for $t\in(0,\infty)$, and $\tilde K^{\gb}(\infty)=1-\int_0^{\infty}\tilde K^{\gb} (t)\dd t$.
Finally, let $\tilde \mu^{\gb}$ be another probability law for $(\tau,\tau')$ defined by:
\begin{itemize}
 \item $\tau_0=0$ $\tilde \mu^{\gb}$-a.s.
 \item $(\tau'_{i+1}-\tau_i)_{i\ge0}$ and $(\tau_i-\tau'_{i})_{i\ge 1}$ are independent i.i.d. sequences,
 \item $\tau'_1$ is an exponential variable of mean $\gl^{-1}$,
 \item $\tau_1-\tau'_1$ has support $\bbR_+\cup\{\infty\}$. On $\bbR_+$, its law
is absolutely continuous w.r.t.\ Lebesgue measure with density $\tilde K^{\gb}(\cdot)$, and $\tilde \mu^{\gb}(\tau_1-\tau'_1=\infty)= \tilde K^{\gb}(\infty)$.
\end{itemize}

Let $\mathcal F_t$ denote the sigma algebra generated by $\cT\cap[0,t]$. We have the following lemma, describing the measure $\mu_{t,\gb}^{\p}$.
\begin{lemma}\label{lemmachin}
 For any $A\in \mathcal F_t$, one has
\begin{equation}\label{hohoho}
 \mu\left[\ind_{A} \, e^{\gb |\cT\cap [0,t]|} \ind_{\{t\in\cT\}}\right]=e^{bt}\tilde \mu^{\gb}(A\cap\{t\in\cT\}).
\end{equation}
As a consequence
\begin{equation}
 \mu_{t,\gb}^{\p}(A):= \tilde \mu^{\gb}(A\ | \ t\in\cT).
\end{equation}
\end{lemma}

\begin{proof}
 
We write  $Z_{t,\gb}^{\p}(A):=\mu\left[\ind_{A}\, e^{\gb |\cT\cap [0,t]|} \ind_{\{t\in\cT\}}\right]$,
and we decompose $Z_{t,\gb}^{\p}(A)$ according to the number of jumps made before $t$. As $A\in \mathcal F_t$, $\ind_A$ can
be written as a function of $(\{\tau_i\ | \tau_i<t\},\{\tau'_i\ | \tau'_i<t\})$ and one has the following integral form for $Z_{t,\gb}^{\p}(A)$,
\begin{eqnarray}
 \lefteqn{Z_{t,\gb}^{\p}(A)= \sum_{n=0}^\infty \int_{0\le t'_1\le t_1\le \dots \le t'_n\le t_n<t}\!\!\!\!\!\!\!\!\!\! \ind_A\,
       e^{(\gb-1)(t-t_n)}\prod_{i=1}^n e^{(\gb-1)(t'_i-t_{i-1})}K(t_i-t'_i) \dd t'_i \dd t_i } \nonumber\\
& = & e^{bt} \sum_{n=0}^{\infty} \int_{0\le t'_1\le t_1\le \dots \le t'_n\le t_n<t}\!\!\!\!\!\!\!\!\!\! \ind_A\,  \gl e^{\gl(t_n-t)}
           \prod_{i=1}^n  e^{\gl(t'_i-t_{i-1})}\tilde K^{\gb}(t_i-t'_i)\dd t'_i \dd t_i     \nonumber\\
&=& e^{bt} \tilde \mu^{\gb}(A\ \cap \ \{ t\in\cT \}).
\label{eq:decompZ}
\end{eqnarray}
\end{proof}

We can now prove some statements from Proposition \ref{annealed},

\begin{proposition}
We have, for $b$ defined as above in \eqref{defb}, $b=\tf(\gb)$. Moreover, 
$b$ can alternatively be defined by
\begin{equation}
\label{youpla}
\int_0^{\infty}e^{-bt} \mu(t\in \cT)\dd t:=\gb^{-1},
\end{equation}
if the equation has a solution and $b=0$ if not.
Let $\gb_c=\inf\{\gb, \tf(\gb)>0\}$, then
\begin{equation}
 \gb_c:=\left(\int_0^{\infty}\mu(t\in \cT)\dd t\right)^{-1}.
\end{equation}
Moreover, if $b>0$, or if $\gb=\gb_c$ and $\int_0^\infty  t K(t)\dd t<\infty$,
then
\begin{equation}\label{krt}
 \lim_{t\to \infty} \tilde \mu^{\gb}(t\in \cT)=\frac{1}{1+\int_0^\infty e^{-bt}t K(t)\dd t}\, .
\end{equation}
\end{proposition}
\begin{rem}\rm
In the case of the homogeneous Random Walk Pinning Model,
one can get the asymptotics of $\tf(\gb)$ around $\gb_c$ given in Proposition \ref{annealed}, by using the local central limit
Theorem for $X_t$ (see \cite[Prop. 7.9, Ch. II]{Spitz} for the
discrete time version, the proof being identical for continuous time). We have
\begin{equation}
 \mu(t\in \cT)=p_t(0)=(cst.+o(1))t^{-d/2}.
\label{eq:lclt}
\end{equation}
Then, Proposition \ref{annealed} follows from \eqref{youpla},
and an application of an Abelian theorem (see \cite[Theorem 2.1]{Book} for the discrete case). 
\label{rem:asympFann}
\end{rem}

\begin{proof}
 We start with the proof of the last item.
 Thanks to the Markov property, one has the following recursion equation 
\begin{multline}
\tilde \mu^{\gb} (t\in \cT)= \tilde\mu^{\gb}(\tau'_1\ge t)+ \int_0^t  \tilde \mu^{\gb}\left(t\in \cT,\ \tau_1\in[s,s+\dd s) \right) \dd s \\
= \exp(-\gl t)+\int_0^t \tilde\mu^{\gb}(\tau_1\in[s,s+\dd s)) \tilde\mu^{\gb} (t-s\in \cT) \dd s.
\end{multline}
By the key renewal theorem, \cite[Theorem 4.7, Ch. V]{asm}, one has
\begin{equation}
 \lim_{t\to\infty}\tilde \mu^{\gb} (t\in \cT):=\frac{\int_0^\infty e^{-\gl t}\dd t}{\gl^{-1}+\int_0^\infty  t \tilde K^{\gb}(t)\dd t}=
  \frac{1}{1+\int_0^{\infty} e^{-bt}t K(t)\dd t}.
\end{equation}
When $b>0$ or $K(t)$ is integrable, the limit is positive. In that case equation \eqref{hohoho} with $A$ equals $\gO$ to the full space gives 
\begin{equation}
 Z_{t,\gb}^{\p}=(cst.+o(1))\exp(bt),
\end{equation}
so that $b=\tf(\gb)$.
For all the other cases, we have necessarily $\tf(\gb)\le 0$ as $Z_{t,\gb}^{\p}\le \gl^{-1}$. To get that $\tf(\gb)= 0$ it is
therefore sufficient to prove that $\tf(\gb)$ is non-negative.
This is done for the random-walk pinning in \cite{BS08}, here it could be done using
the assumption $\int_0^{\infty} e^{\gep t}K(t)\dd t=\infty$ for all $\gep$ (which is also necessary).

\medskip

Now, we turn to the proof of \eqref{youpla}.
Let $K_1(t)=e^{-t}$ be the density with respect to the Lebesgue measure of $\tau'_1$ (under $\mu$).
For $t> 0$,
\begin{equation}
 \mu(t\in \cT)=e^{-t}+\int_0^t \sum_{n=1}^{\infty} \mu(\tau_n\in [s,s+\dd s)) e^{-(t-s)}\dd s=\sum_{n=0}^{\infty} [(K_1 \ast K)^{\ast n}\ast K_1] (t).
\end{equation}
Therefore using the fact that Laplace transform transforms convolutions into products, one obtains, for all $b>0$

\begin{eqnarray}
 \int_0^{\infty} e^{-bt} \mu(t\in \cT)\dd t &=& \sum_{n=0}^{\infty} \left(\int_0^{\infty} e^{-(b+1)t}\dd t\right)^{n+1}\left(\int_0^{\infty}e^{-bt}K(t)\dd t\right)^n\nonumber\\
  &=& \frac{1}{1+b}\frac{1}{1-\frac{1}{1+b}\int_0^{\infty}e^{-bt} K(t)\dd t},
\end{eqnarray}
which with \eqref{defb} gives us the right result (the case $\tf(\gb)=0$ is obtained by continuity and non-negativity of the free-energy).
The value of $\gb_c$ is then an easy consequence.
\end{proof}
We now give a Corollary that describes the local intersection time $L_t(X,0)$ under $\mu_{t,\gb}^{\p}$.
\begin{cor}
When $b>0$ or when $\gb=\gb_c$ and $\int_0^{\infty} t K(t)\dd t<\infty$,
$\frac{|\cT\cap [0,t]|}{t}$ under $\mu_{t,\gb}^{\p}$ converges in probability to
\begin{equation}
 \frac{1}{1+\int_0^{\infty} e^{-bt} t K(t)\dd t} >0 .
\end{equation}
 \label{cor:purecontact}
\end{cor}

\begin{proof}
As $\int_0^{\infty} t K(t)\dd t<\infty$, the law of large numbers (applied first for the renewal
process $\tau$ and then to the sum of independent exponential times) tells us that
\begin{equation}
\lim_{t\to \infty} \frac{|\cT\cap [0,t]|}{t}=\frac{1}{1+\int_0^\infty e^{-bt} t K(t) \dd t}, \indent  \tilde \mu^{\gb}-a.s,
\end{equation}
and therefore the convergence also holds in probability.
 
Restricted on $\mathcal F_t$, the measure $\mu_{t,\gb}^{\p}$ is equal to $\tilde\mu^{\gb}(\cdot \ | t\in \cT)$ and
we also have that $\tilde  \mu^{\gb}(t\in \cT)$ is bounded away from zero by \eqref{krt}. This gives us that
the law of $\frac{|\cT\cap [0,t]|}{t}$ under $\mu_{t,\gb}^{\p}$ converges in probability to the same limit.
\end{proof}

\begin{rem}\rm \label{temloc}
In dimension $d$, as noted in Remark \ref{rem:asympFann}, the local central limit Theorem gives
\begin{equation}\label{tail}
\mu(t\in \cT)= \bbP^X(X_t=0)=(cst.+o(1))t^{-d/2}.
\end{equation}
For $d\ge 5$ this implies that $\int_0^{\infty} t K(t)\dd t<\infty$. Indeed for $t$ large enough
\begin{equation}
 \bbP^X(X_t=0)\ge \int_{0<t_1<t_2<t} e^{-t_1}K(t_2-t_1)e^{-(t-t_2)}\dd t_1 \dd t_2\ge cst. \int_{t-2}^{t-1} K(s) \dd s,
\end{equation}
so that $\int_{t-2}^{t-1} K(s)\dd s =O(t^{-d/2})$. Therefore, one can apply Corollary \ref{cor:purecontact} to get that $L_t(X,0)$ is of order $t$ for $\gb=\gb_c$.
\end{rem}

We present here an advanced version of \eqref{tralala}, which was used for the proof of the upper bound in Theorem \ref{lowtemp}.

\begin{lemma}\label{cocorico}
 For any value of $t$, for any random walk $X$ with jump rate $1$ one has
\begin{equation}
 e^{(\gb-1)t} \le \bbE^X\left[e^{L_t(X,0)}\right]\le e^{\left(\gb-1+\frac{1}{\gb}\right)t}\left(1+\frac{1}{\gb}\right).
\end{equation}
\end{lemma}
\begin{proof}
The left hand side inequality is simply obtained by considering the contribution of trajectories that never jumps.
To obtain the other inequality we decompose the partition function according to the time $t'_i$, $t_i$ that are respectively the $i$-th jump out of zero,
and the $i$-th return to zero. We write $Z_{t,\gb}=\bbE^X\left[e^{L_t(X,0)}\right]$,
\begin{multline}
 Z_{t,\gb}=e^{(\gb-1)t}\sum_{n=0}^{\infty}\int_{0<t'_1<t_1<\dots<t_n<t}\prod_{i=1}^n e^{-\gb(t'_i-t_{i-1})} K(t_i-t'_i) \dd t'_i \dd t_i \\
\left[1+\int_{t_n}^t e^{-\gb(t-t'_{n+1})}\dd t'_{n+1} \int_t^{\infty} K(t_{n+1}-t'_{n+1})\dd t_{n+1}\right]. 
\end{multline}
Then one remark that for the random walk $K(t)\le 1$ for all $t$. Indeed the probability that after the
first jump, the first excursion returns within a time in the interval $[t,t+\dd t]$ (which is equal to $K(t)\dd t$) is smaller
than the probability that $X$ makes a jump in the interval $[t,\dd t]$ (which is equal to $\dd t$). Hence

\begin{equation}\begin{split}
\int_{t_{i-1}}^{t_{i}}e^{-\gb(t'_i-t_{i-1})} K(t_i-t'_i)\dd t'_i &\le \frac{1}{\gb},\\
\int_{t_n}^t e^{-\gb(t-t'_{n+1})}\dd t'_{n+1}&\le \frac{1}{\gb}\\
\int_t^{\infty} K(t_{n+1}-t'_{n+1})\dd t_{n+1}&\le 1.
\end{split}\end{equation}
Therefore
\begin{eqnarray}
  Z_{t,\gb}&\le& e^{(\gb-1)t}\sum_{n=0}^{\infty}\int_{0<t_1<\dots<t_n}\frac{1}{\gb^n}\left(1+\frac{1}{\gb}\right)\dd t_1\dots\dd t_n \nonumber\\
  &= & e^{(\gb-1)t}\left(1+\frac{1}{\gb}\right)\sum_{n=0}^{\infty} \frac{t^n}{\gb^n n!}
\end{eqnarray}
which is exactly the result.
\end{proof}

{\bf Acknowledgements}: 
The authors would like to thank F.L.\ Toninelli for his constant support in this project and his precious advice.
Q.B.\ is very grateful to the Mathematics Department of Università di Roma Tre for hosting him while working on this project.
H.L. acknowledges the support of ECR grant PTRELSS.

\end{document}